\documentclass[journal]{IEEEtran}
\usepackage{amsmath,amsfonts,amssymb}
\usepackage{mathrsfs}
\usepackage{algorithmic}
\usepackage{algorithm}
\usepackage{array}
\usepackage{textcomp}
\usepackage{stfloats}
\usepackage{url}
\usepackage{hyperref}
\usepackage{dsfont}
\hypersetup{colorlinks=true,linkcolor=blue,filecolor=blue,urlcolor=blue,citecolor=blue}
\usepackage{xcolor}
\usepackage{verbatim}
\usepackage{graphicx}
\usepackage{tikz}
\usetikzlibrary{shapes.geometric, arrows, positioning, fit, calc, backgrounds}
\graphicspath{{Figures/}}
\usepackage{epstopdf}
\usepackage{cite}
\usepackage{multirow}
\usepackage{bm}
\usepackage{siunitx}
\usepackage{booktabs}
\newtheorem{remark}{Remark}
\newtheorem{assumption}{Assumption}
\newtheorem{theorem}{Theorem}
\newtheorem{corollary}{Corollary}
\usepackage{flushend}
\def\BibTeX{{\rm B\kern-.05em{\sc i\kern-.025em b}\kern-.08em T\kern-.1667em\lower.7ex\hbox{E}\kern-.125emX}}
\begin{document}

\title{Effective Depth in Joint Source-Channel Coding: An Implicit Equilibrium Analysis}

\author{Kaiwen Yu,~\IEEEmembership{Member,~IEEE},
	Gang Wu,~\IEEEmembership{Senior Member,~IEEE},	
	Xiaodong Xu,~\IEEEmembership{Senior Member,~IEEE},
	Yi Ma,~\IEEEmembership{Senior Member,~IEEE},
	and Rahim Tafazolli,~~\IEEEmembership{Fellow,~IEEE}

	\thanks{This work has been submitted to the IEEE for possible publication. Copyright may be transferred without notice, after which this version may no longer be accessible.}
	
	\thanks{Kaiwen Yu and Gang Wu are with the National Key Laboratory of Wireless Communications, University of Electronic Science and Technology of China, Chengdu 611731, China. Gang Wu is also with the Tianfu Jiangxi Laboratory, Chengdu, China (e-mail: yukaiwen@uestc.edu.cn; wugang99@uestc.edu.cn).}
	
	\thanks{Xiaodong Xu is with the State Key Laboratory of Networking and Switching Technology, Beijing University of Posts and Telecommunications, Beijing 100876, China. He is also with the Department of Broadband Communication, Peng Cheng Laboratory, Shenzhen 518055, China, and the Tianfu Jiangxi Laboratory, Chengdu, China (e-mail: xuxiaodong@bupt.edu.cn).}
	
	\thanks{Yi Ma and Rahim Tafazolli are with 5GIC \& 6GIC, Institute for Communication Systems (ICS), University of Surrey, GU2 7XH Guildford, U.K. (e-mail: y.ma@surrey.ac.uk; r.tafazolli@surrey.ac.uk).}
	
}

\maketitle

\begin{abstract}
	A fundamental design question in deep joint source-channel coding (Deep JSCC) remains insufficiently explored: given a channel signal-to-noise ratio (SNR), what effective computation depth is required for semantic reconstruction?
	Existing Deep JSCC systems typically employ fixed-depth neural architectures selected through empirical hyperparameter tuning, which may lead to unnecessary computation under favorable channel conditions and insufficient refinement under severe channel noise.
	This paper proposes \emph{Implicit-JSCC}, an implicit equilibrium framework in which semantic encoding and decoding are formulated as fixed-point equilibrium processes.
	The effective encoder and decoder depths are determined by residual-based solver convergence rather than manually predefined layer numbers, while parameter sharing across equilibrium iterations enables depth-independent parameter complexity.
	To analyze the resulting effective-depth behavior, we develop a Gaussian-process-inspired kernel evolution framework that models equilibrium iterations as an effective-depth propagation process.
	Since channel noise is injected between the encoder and decoder, the analysis tracks channel-induced representation perturbations across receiver-side equilibrium iterations and derives a theory-guided depth--SNR relationship.
	After offline calibration of the system-specific parameters, the resulting model characterizes the required receiver-side refinement depth under different SNRs.
	Extensive experiments show that Implicit-JSCC achieves competitive reconstruction performance while enabling residual-based adaptive inference and controllable computation--quality tradeoffs.
	The depth--SNR model further provides a characterization of the SNR-dependent refinement depth required to reach a prescribed perturbation tolerance.
\end{abstract}


\section{Introduction}

\IEEEPARstart{T}{he} rapid evolution of sixth-generation (6G) wireless networks is driving a fundamental transition from conventional bit-level communication toward semantic communication, where the objective is no longer merely reliable bit delivery but the efficient exchange of task-relevant semantic information~\cite{weaver1953recent}.
This paradigm is particularly relevant to AI-agent communication scenarios, where visual observations must be exchanged under limited bandwidth and time-varying wireless channels.
Among existing semantic communication paradigms, deep joint source-channel coding (Deep JSCC)~\cite{bourtsoulatze2019deep} has emerged as a promising framework by replacing the conventional separate source and channel coding pipeline with an end-to-end neural architecture.
Recent studies have further improved the performance and robustness of Deep JSCC through advanced neural architectures and adaptive representation learning mechanisms~\cite{xu2022adjscc,yang2025swinjscc,yu2024two,yu2025partial}.

Despite these advances, as shown in Fig.~\ref{fig:framework}(a), existing Deep JSCC systems usually rely on fixed-depth encoder--decoder architectures, where the number of neural layers is manually selected through empirical hyperparameter tuning. However, the required effective depth should depend on channel conditions. In high-SNR regimes, excessive depth may introduce redundant computation after semantic representations have already been sufficiently recovered. In contrast, under low-SNR regimes, additional refinement may be required to suppress channel-induced perturbations. This motivates a fundamental design question:
\begin{quote}
	\emph{Given a channel SNR, what effective computation depth is required for semantic reconstruction?}
\end{quote}

This question calls for a JSCC framework that can adapt its effective computation depth to channel conditions.
Such adaptability should also be supported by an analytical framework that explains how channel-induced perturbations are progressively refined across effective iterations.
Existing studies primarily focus on architectural design and empirical performance optimization, while the relationship among effective depth, representation refinement, and channel SNR remains insufficiently understood.

To address this issue, as shown in Fig.~\ref{fig:framework}(b), this paper introduces an implicit equilibrium framework for JSCC based on deep equilibrium models (DEQs)~\cite{bai2019deep}.
Instead of explicitly stacking multiple neural layers, the proposed framework repeatedly applies a shared transformation across equilibrium iterations, so that the parameter complexity remains independent of the effective depth.
Furthermore, we develop a neural-network Gaussian process (NNGP)-inspired kernel evolution analysis~\cite{lee2018deep} to characterize the depth--SNR behavior of the proposed implicit framework.
Based on the NNGP kernel-evolution perspective, we establish a theory-guided depth--SNR analytical model that relates channel-induced representation perturbations to the required equilibrium iterations through an effective contraction characterization.
Since channel noise is injected between the encoder and decoder, the analytical characterization focuses on receiver-side perturbation propagation.
After calibrating the system-specific parameters, the resulting model characterizes the required receiver-side refinement depth under different SNRs.
The main contributions of this paper are summarized as follows:
\begin{itemize}
	
	\item
	We propose \emph{Implicit-JSCC}, an implicit equilibrium JSCC framework that reformulates semantic encoding and decoding as fixed-point equilibrium processes.
	Different from conventional Deep JSCC architectures with manually predefined layer numbers, the proposed framework determines the effective computation depth through solver convergence while sharing parameters across equilibrium iterations, thereby achieving parameter complexity that is independent of the executed effective depth.
	
	\item
	We develop a Gaussian-process-inspired kernel evolution analysis to characterize the effective-depth behavior of the proposed implicit JSCC framework under channel-induced perturbations.
	Based on the NNGP kernel-evolution perspective, we establish a theory-guided depth--SNR analytical model that relates channel-induced representation perturbations and the required equilibrium iterations.
	After calibrating the system-specific parameters, the resulting model provides an SNR-dependent characterization of the required refinement depth.
	
	\item
	Experiments on DIV2K and Kodak24 show that Implicit-JSCC achieves competitive reconstruction quality, supports residual-based adaptive inference under varying SNRs, and mitigates both insufficient refinement caused by too few iterations and redundant computation caused by unnecessarily large fixed-depth budgets.
\end{itemize}

The remainder of this paper is organized as follows.
Section II reviews related work.
Section III presents the proposed Implicit-JSCC architecture and residual-based adaptive solver.
Section IV develops the NNGP-inspired depth--SNR analytical framework and calibrated effective-depth model.
Section V reports the experimental results.
Section VI concludes this paper.
Table~\ref{tab:notations} summarizes the main symbols used throughout this paper.

\begin{figure}[!t]
	\centering
	\includegraphics[width=\columnwidth]{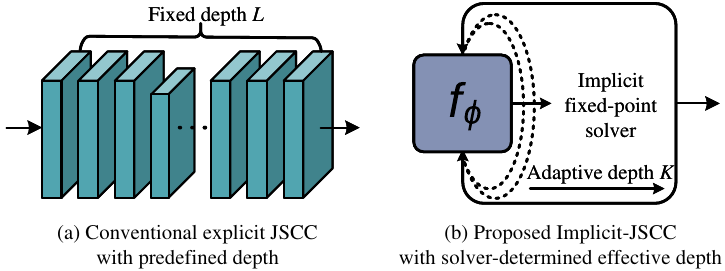}
	\caption{Conceptual comparison between conventional explicit JSCC with predefined depth and the proposed Implicit-JSCC with solver-determined effective depth.}
	\label{fig:framework}
\end{figure}

\begin{table}[!t]
	\centering
	\caption{Main notations}
	\label{tab:notations}
	\renewcommand{\arraystretch}{1.2}
	\begin{tabular}{@{}cl@{}}
		\toprule
		Symbol & Description \\
		\midrule
		$\mathbf{s}$ & Source image \\
		$\mathbf{x}$ & Transmitted channel symbols \\
		$\mathbf{y}$ & Received noisy symbols \\
		$R$ & Bandwidth ratio \\
		$\mathbf{z}^{*}$ & Equilibrium representation \\
		$k$ & Fixed-point solver iteration index \\
		$K_{\mathrm{enc}}, K_{\mathrm{dec}}$ & Effective encoder and decoder iteration counts \\
		$K_{\mathrm{dec}}^{\mathrm{pred}}$ & Theory-guided predicted decoder effective depth \\
			$K_{\mathrm{dec}}^{\mathrm{fix}}$ & Preset decoder iteration budget for fixed-depth inference \\
			$\mathcal{K}^{(k)}$ & NNGP covariance kernel at solver iteration $k$ \\
			$\mathcal{J}_{\star}$ & Local linearized kernel-propagation operator \\
			$\sigma_n^2$ & Channel noise variance \\
		$\alpha$ & Damping coefficient in the residual implicit update \\
		$\beta$ & Effective kernel propagation factor\\
		$\tau$ & Generic solver stopping threshold \\
		$\epsilon_{\mathrm{tol}}$ & Kernel-space convergence tolerance\\
		$\tau_{\mathrm{enc}},\tau_{\mathrm{dec}}$ & Encoder and decoder stopping thresholds \\
		$c_0,c_1$ & Residual and noise-dependent perturbation coefficients \\
		$\eta=c_1/c_0$ & Relative noise-scale ratio \\
		$\mathcal{V}$ & Calibration set \\
		$\mathcal{S},\mathcal{S}_{\mathrm{H}}$ & Calibration SNR grid and its high-SNR subset \\
		$\overline{K}_{\mathrm{dec}}(s)$ & Mean measured decoder depth at SNR $s$ \\
		$\overline{K}_{\mathrm{H}}$ & Mean measured depth over $\mathcal{S}_{\mathrm{H}}$ \\
		$\hat{\beta},\hat{\eta},\hat{c}_0,\hat{c}_1$ & Calibrated parameter estimates \\
		\bottomrule
	\end{tabular}
\end{table}

\section{Related Work}\label{s2}

\subsection{Deep JSCC and Semantic Communication}

Recent JSCC and semantic communication studies have improved reconstruction quality and robustness through more expressive neural architectures~\cite{nguyen2026contemporary,zhang2026comai}.
Transformer-based and multi-scale vision models have been introduced to enhance semantic feature extraction and robustness against channel impairments~\cite{peng2025robust}, while state space models have been explored to capture long-range dependencies with reduced computational complexity~\cite{MambaJSCC}.
Generative models have also been investigated\cite{fan2026generative}, including semantic image transmission based on a variational autoencoder--generative adversarial network (VAE-GAN)~\cite{omi2026vaegan} and diffusion-based JSCC for variable-rate multiple-input multiple-output (MIMO) channel state information (CSI) reconstruction~\cite{liu2026residual}.
Beyond image reconstruction, semantic communication has been extended to task-oriented and multimodal scenarios, such as knowledge-base-assisted transmission, scene-graph-aided image communication, cross-modal semantic generation, speech translation, and semantic video transmission~\cite{zhu2026scene,li2026universal,chen2026crossmodal,liu2026speech,gao2026volumetric,wang2026surveillance}.
Other studies have investigated semantic reliability and practical deployment issues, including semantic hybrid automatic repeat request (HARQ), coverage enhancement, distributed semantic systems for heterogeneous Internet of Things (IoT) devices, and free-space optical semantic communication~\cite{an2026spharq,wang2026coverage,zeng2026distributed,wei2026fso,yu2025distributed}.
These works demonstrate the broad applicability of semantic communication in wireless networks, but most of them improve representation capability or transmission robustness within explicitly predefined neural architectures.

Existing adaptive JSCC methods further improve robustness by adjusting feature importance, attention weights, modulation strength, or coding rate according to channel conditions~\cite{xu2022adjscc}.
However, their neural transceiver architectures are usually still explicitly stacked, so once the encoder and decoder backbones are selected, the number of computational layers remains fixed during inference.
Therefore, existing adaptive JSCC methods mainly adapt feature representation or transmission strategy, rather than the effective computation depth of the semantic transceiver.
In contrast, this paper focuses on effective-depth inference.
Instead of merely increasing backbone capacity or using a fixed-depth neural transceiver for all channel conditions, the proposed Implicit-JSCC framework formulates semantic encoding and decoding as fixed-point equilibrium processes.
The effective encoder and decoder depths are determined by solver convergence, enabling adaptive semantic refinement without manually specifying the number of layers.

\subsection{Theoretical Analysis of Semantic Communication}

Early efforts beyond Shannon's classical framework aimed to define and quantify semantic information by considering meaning, synonymity, task relevance, and knowledge representation~\cite{weaver1953recent}.
Information-theoretic tools have further been used to formulate semantic communication objectives, such as information bottleneck-based semantic representation learning~\cite{zhang2026beyond,liu2026rateib,zhang2026rib,wei2026taskagnostic}, conditional rate-distortion formulations for task-adaptive transmission~\cite{he2026taskadaptive}, and information bottleneck or rate-splitting methods for broadcast and multi-user semantic communication~\cite{xu2026broadcast}.
Robustness, security, and adaptability have also been investigated through robust information bottleneck methods, secure semantic transmission, physical-layer protection, feedback-assisted adaptation, knowledge base updates, and foundation model-assisted adaptation~\cite{zhang2026rib,zeng2026distributed,li2026deepguard,zhang2026skb,he2026taskadaptive,liu2026foundation}.
However, these studies mainly focus on rate, distortion, semantic relevance, task utility, reliability, or system-level adaptation, rather than the effective computation depth of neural semantic transceivers under channel noise.

Implicit layers and deep equilibrium models formulate neural representations as fixed points of nonlinear transformations, enabling parameter sharing across effective iterations and avoiding explicit layer stacking~\cite{bai2019deep}.
Such models are attractive for effective-depth inference because the number of refinement steps can be determined by solver convergence rather than by a manually predefined network depth.
However, their potential for wireless semantic communication remains underexplored, especially under channel-dependent perturbations.
In parallel, neural network Gaussian process (NNGP) theory provides an analytical tool for studying infinitely wide neural networks through recursively evolving covariance kernels~\cite{lee2018deep}.
Since NNGP characterizes how covariance kernels evolve with network depth, it offers a useful perspective for studying representation propagation and perturbation decay across effective-depth iterations.
Nevertheless, existing semantic communication theories have not explicitly connected channel SNR, kernel-space representation perturbation, and effective-depth inference in an implicit JSCC framework.

This paper addresses this gap by developing an NNGP-inspired analytical characterization for implicit-JSCC.
Unlike existing semantic communication theories that mainly focus on rate-distortion, information bottleneck, or task utility, our analysis studies how channel-induced perturbations affect receiver-side representation refinement across effective-depth iterations.
Together with the implicit equilibrium transceiver design, this analysis provides a theory-guided explanation of the depth--SNR relationship and supports residual-based effective-depth inference for semantic transmission.

\section{Implicit Equilibrium JSCC Architecture}\label{s3}

\subsection{JSCC System Overview}

Consider an image source $\mathbf{s}\in\mathbb{R}^{H\times W\times C}$, where $H$, $W$, and $C$ denote the image height, width, and number of color channels, respectively.
The transmitter directly maps the source image into channel symbols through a neural encoder, and the receiver reconstructs the image from the noisy channel observations through a neural decoder.

The transmitter generates a real-valued tensor
\begin{equation}
	\mathbf{x}
	=
	E_{\bm{\theta}}(\mathbf{s}).
\end{equation}
Here, $E_{\bm{\theta}}(\cdot)$ denotes the overall neural encoder mapping, which is later instantiated by a convolutional stem, an implicit equilibrium encoder, and an output head.
The corresponding complex-valued representation is denoted by
\(\mathbf{x}_{\mathrm c}\in\mathbb{C}^{{C_c\times H'\times W'}}\), where $C_c$, $H'$, and $W'$ denote the channel-symbol dimension and spatial dimensions.
The bandwidth ratio is defined as
$
	R
	=
	\frac{C_cH'W'}{CHW},
$
which measures the number of transmitted complex channel symbols per source pixel. 
The transmitted symbols are power-normalized to satisfy
$
	\mathbb{E}\left[|x_{{\mathrm c},i}|^2\right]=1.
$
The wireless channel is modeled as
\begin{equation}\label{eq:wireless_channel}
	\mathbf{y}_{\mathrm c}
	=
	\mathbf{h}\odot\mathbf{x}_{\mathrm c}
	+
	\mathbf{n}_{\mathrm c},
\end{equation}
where $\mathbf{h}$ denotes the complex channel coefficient tensor, $\odot$ denotes element-wise multiplication, and 
$\mathbf{n}_{\mathrm c}\sim\mathcal{CN}(\mathbf{0},\sigma_n^2\mathbf{I})$
is additive complex Gaussian noise.
The AWGN case is recovered by setting $\mathbf{h}=\mathbf{1}$.
For Rayleigh fading, the entries of $\mathbf{h}$ follow $\mathcal{CN}(0,1)$.
For Rician fading, the entries of $\mathbf{h}$ follow $\mathcal{CN}(\mu,\sigma_h^2)$ with factor $\kappa=|\mu|^2/\sigma_h^2$.
For Rayleigh and Rician fading, a minimum mean-squared error (MMSE) equalizer is applied before neural decoding.

The complex received signal $\mathbf{y}_{\mathrm c}$ is converted back to its real-valued representation $\mathbf{y}$ by separating the in-phase and quadrature components. 
The neural decoder then reconstructs the source image as
\begin{equation}
	\hat{\mathbf{s}}
	=
	D_{\bm{\varphi}}(\mathbf{y}).
\end{equation}
Here, $D_{\bm{\varphi}}(\cdot)$ denotes the overall neural decoder mapping, which is later instantiated by an input projection layer, an implicit equilibrium decoder, and a reconstruction head.

\begin{figure*}[t]
	\centering
	\includegraphics[width=0.95\textwidth]{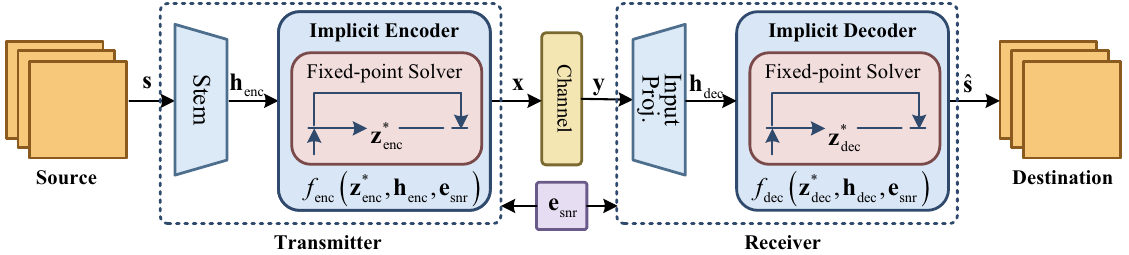}
	\caption{System architecture of the proposed Implicit-JSCC framework. The encoder and decoder are formulated as implicit equilibrium modules, where a shared fixed-point block is repeatedly applied until the residual-based stopping criterion is satisfied.}
	\label{fig:systemmodel}
\end{figure*}

As illustrated in Fig.~\ref{fig:systemmodel}, the proposed Implicit-JSCC replaces explicitly stacked neural layers with implicit equilibrium modules.
The effective encoder and decoder depths, denoted by $K_{\mathrm{enc}}$ and $K_{\mathrm{dec}}$, are determined by fixed-point solver convergence rather than manually specified layer numbers.
The detailed transmitter and receiver operations are described in the following subsections.

Intuitively, an implicit equilibrium module can be viewed as repeatedly refining a latent representation until the representation becomes self-consistent under a shared transformation. Unlike an explicitly stacked network, where each layer has independent parameters and the depth is fixed after architecture design, the implicit module reuses the same transformation and lets the solver determine how many refinement steps are executed for a given input and channel condition.

\subsection{Implicit Encoder--Decoder Transceiver}

At the transmitter, the source image
$
\mathbf{s}\in\mathbb{R}^{H\times W\times C}
$
is first projected into low-dimensional semantic features by a convolutional stem:
\begin{equation}
	\mathbf{h}_{\mathrm{enc}}
	=
	\mathrm{Stem}_{\mathrm{enc}}(\mathbf{s}).
\end{equation}

The channel SNR is embedded through sinusoidal encoding followed by a lightweight multilayer perceptron (MLP):
\begin{equation}
	\mathbf{e}_{\mathrm{snr}}
	=
	\mathrm{MLP}
	\left(
	\mathrm{SinEmbed}
	\left(
	\mathrm{SNR}_{\mathrm{dB}}
	\right)
	\right).
\end{equation}
The embedding $\mathbf{e}_{\mathrm{snr}}$ provides channel-state conditioning for the implicit encoder and decoder.

Instead of passing $\mathbf{h}_{\mathrm{enc}}$ through a predefined number of neural layers, the transmitter obtains an encoder equilibrium representation by solving
\begin{equation}\label{eq:enc_equilibrium}
	\mathbf{z}_{\mathrm{enc}}^{*}
	=
	f_{\mathrm{enc}}
	\left(
	\mathbf{z}_{\mathrm{enc}}^{*},
	\mathbf{h}_{\mathrm{enc}},
	\mathbf{e}_{\mathrm{snr}};
	\bm{\psi}_{\mathrm{enc}}
	\right),
\end{equation}
where $f_{\mathrm{enc}}(\cdot)$ is the shared implicit encoder update and $\bm{\psi}_{\mathrm{enc}}$ denotes its parameters. 
The encoder parameters are shared across equilibrium iterations, while the encoder and decoder use separate parameter sets.
The equilibrium representation is then mapped into channel symbols:
$
	\mathbf{x}
	=
	\mathrm{Head}_{\mathrm{enc}}
	\left(
	\mathbf{z}_{\mathrm{enc}}^{*}
	\right).
$

At the receiver, the noisy observation $\mathbf{y}$ is first projected into received semantic features:
\begin{equation}
	\mathbf{h}_{\mathrm{dec}}
	=
	\mathrm{InputProj}(\mathbf{y}).
\end{equation}

The decoder then refines the noisy received representation through an implicit equilibrium process:
\begin{equation}\label{eq:dec_equilibrium}
	\mathbf{z}_{\mathrm{dec}}^{*}
	=
	f_{\mathrm{dec}}
	\left(
	\mathbf{z}_{\mathrm{dec}}^{*},
	\mathbf{h}_{\mathrm{dec}},
	\mathbf{e}_{\mathrm{snr}};
	\bm{\psi}_{\mathrm{dec}}
	\right),
\end{equation}
where $f_{\mathrm{dec}}(\cdot)$ is the shared implicit decoder update and $\bm{\psi}_{\mathrm{dec}}$ denotes its parameters. 
The decoder equilibrium representation $\mathbf{z}_{\mathrm{dec}}^{*}$ is determined by the received noisy features and the SNR condition.
Finally, the reconstructed image is obtained by
$
	\hat{\mathbf{s}}
	=
	\mathrm{Head}_{\mathrm{dec}}
	\left(
	\mathbf{z}_{\mathrm{dec}}^{*}
	\right).
$

\subsection{Shared Implicit Update and Adaptive Solver}

Although the transmitter and receiver use different parameter sets, their implicit modules follow the same generic fixed-point form:
\begin{equation}\label{eq:generic_equilibrium}
	\mathbf{z}_{q}^{*}
	=
	f_q
	\left(
	\mathbf{z}_{q}^{*},
	\mathbf{h}_{q},
	\mathbf{e}_{\mathrm{snr}};
	\bm{\psi}_{q}
	\right),
	\quad
	q\in\{\mathrm{enc},\mathrm{dec}\},
\end{equation}
where $\mathbf{h}_{q}$ corresponds to the encoder feature $\mathbf{h}_{\mathrm{enc}}$ for $q=\mathrm{enc}$ and to the received feature $\mathbf{h}_{\mathrm{dec}}$ for $q=\mathrm{dec}$.

\subsubsection{Contraction-Oriented Implicit Update}
\label{sec:block}

The implicit update is designed to promote stable fixed-point convergence while retaining an explicit conditioning path from the encoder or decoder input feature.
In the implemented model, the implicit mapping is instantiated by a residual convolutional DEQ block.
Specifically, the conditioning feature $\mathbf{h}_q$ is projected to $\mathbf{u}_q$, and the fixed-point map is written in the compact form
\begin{equation}\label{eq:block}
	f_q(\mathbf{z}_{q})
	=
	(1-\alpha_q)\mathbf{z}_{q}
	+
	\alpha_q
	\left[
	\mathbf{u}_{q}
	+
	\Delta_q
	\left(
	\mathbf{z}_{q},
	\mathbf{u}_{q},
	\mathbf{e}_{\mathrm{snr}}
	\right)
	\right],
\end{equation}
where $\Delta_q(\cdot)$ is implemented by a residual convolutional module with normalization, SNR-conditioned modulation, and residual injection.
This update can be interpreted as a damped interpolation between the current equilibrium state $\mathbf{z}_q$ and an SNR-conditioned residual target $\mathbf{u}_q+\Delta_q(\cdot)$.
The damping coefficient $\alpha_q\in(0,1)$ is learnable and bounded in implementation.

\textit{Proposition 1 (Sufficient contraction condition):}
For fixed $\mathbf{u}_q$ and $\mathbf{e}_{\mathrm{snr}}$, assume that
$
g_q(\mathbf{z}_{q})
=
\mathbf{u}_{q}
+
\Delta_q(\mathbf{z}_{q},\mathbf{u}_{q},\mathbf{e}_{\mathrm{snr}})
$
is $L_g$-Lipschitz continuous with respect to $\mathbf{z}_{q}$, i.e.,
\begin{equation}
	\left\|
	g_q(\mathbf{z}_1)
	-
	g_q(\mathbf{z}_2)
	\right\|
	\le
	L_g
	\left\|
	\mathbf{z}_1-\mathbf{z}_2
	\right\|.
	\label{eq:g_lipschitz}
\end{equation}
Then the damped update
$
f_q(\mathbf{z}_{q})
=
(1-\alpha_q)\mathbf{z}_{q}
+
\alpha_q g_q(\mathbf{z}_{q})
$
satisfies
\begin{equation}\label{eq:lipschitz_general}
	\left\|
	f_q(\mathbf{z}_1)
	-
	f_q(\mathbf{z}_2)
	\right\|
	\le
	\left(
	1-\alpha_q+\alpha_q L_g
	\right)
	\left\|
	\mathbf{z}_1-\mathbf{z}_2
	\right\|.
\end{equation}
Therefore, $f_q$ is contractive if
$
1-\alpha_q+\alpha_q L_g<1
$,
which is equivalent to $L_g<1$ for $\alpha_q\in(0,1)$.

\textit{Proof:}
For any two states $\mathbf{z}_1$ and $\mathbf{z}_2$, the difference between their damped updates is
\begin{align}
	f_q(\mathbf{z}_1)-f_q(\mathbf{z}_2)
	&=
	(1-\alpha_q)(\mathbf{z}_1-\mathbf{z}_2) \notag\\
	&\quad+
	\alpha_q
	\left[
	g_q(\mathbf{z}_1)-g_q(\mathbf{z}_2)
	\right].
	\label{eq:f_difference}
\end{align}
Taking the norm on both sides and applying the triangle inequality gives
\begin{align}
	\|f_q(\mathbf{z}_1)-f_q(\mathbf{z}_2)\|
	&\le
	(1-\alpha_q)\|\mathbf{z}_1-\mathbf{z}_2\| \notag\\
	&\quad
	+
	\alpha_q\|g_q(\mathbf{z}_1)-g_q(\mathbf{z}_2)\|.
	\label{eq:triangle_bound}
\end{align}
Using the $L_g$-Lipschitz continuity of $g_q$ in~\eqref{eq:g_lipschitz}, we further obtain
\begin{align}
	\left\|
	f_q(\mathbf{z}_1)-f_q(\mathbf{z}_2)
	\right\|
	&\le
	(1-\alpha_q)
	\left\|
	\mathbf{z}_1-\mathbf{z}_2
	\right\|
	+
	\alpha_q L_g
	\left\|
	\mathbf{z}_1-\mathbf{z}_2
	\right\| \notag\\
	&=
	\left(
	1-\alpha_q+\alpha_q L_g
	\right)
	\left\|
	\mathbf{z}_1-\mathbf{z}_2
	\right\|.
\end{align}
This proves~\eqref{eq:lipschitz_general}. 
For $f_q$ to be contractive, the Lipschitz upper bound should be smaller than one, i.e.,
$
	1-\alpha_q+\alpha_q L_g<1 .
$
Since $\alpha_q\in(0,1)$, this condition is equivalent to
$
\alpha_q(L_g-1)<0
$,
and hence reduces to $L_g<1$.
\hfill $\blacksquare$

This result indicates that damping alone does not guarantee contraction when $L_g>1$.
In practice, the implemented implicit block combines bounded damping and normalization, while Jacobian regularization is used during training to promote, rather than guarantee, stable fixed-point behavior.

\subsubsection{Adaptive Fixed-Point Solver}

The equilibrium state is obtained through fixed-point iterations:
\begin{equation}\label{eq:solver}
	\mathbf{z}_{q}^{(k+1)}
	=
	f_q
	\left(
	\mathbf{z}_{q}^{(k)},
	\mathbf{h}_{q},
	\mathbf{e}_{\mathrm{snr}};
	\bm{\psi}_q
	\right),
	\quad
	q\in\{\mathrm{enc},\mathrm{dec}\}.
\end{equation}

The solver computes the mixed absolute--relative residual
\begin{equation}\label{eq:stopping}
	r_{q}^{(k)}
	=
	\frac{
		\|
		\mathbf{z}_{q}^{(k+1)}
		-
		\mathbf{z}_{q}^{(k)}
		\|
	}{
		\max
		\left\{
		\|
		\mathbf{z}_{q}^{(k)}
		\|,
		1
		\right\}
	},
	\quad
	q\in\{\mathrm{enc},\mathrm{dec}\},
\end{equation}
where the unit lower bound makes the criterion absolute when the state norm is small and relative when it is large.
The solver stops when $r_q^{(k)}<\tau_q$ or when the maximum iteration number $K_{\max}$ is reached.
The residual $r_q^{(k)}$ measures how much the latent representation changes after one additional shared-block update. A small residual indicates that the representation has nearly reached a fixed point, and the corresponding iteration count is interpreted as the executed effective depth.

The effective depth of module $q$ is therefore defined as
\begin{equation}\label{eq:effective_depth}
	K_q
	=
	\min
	\left\{
	k+1:
	r_q^{(k)}<\tau_q
	\right\},
	\quad
	q\in\{\mathrm{enc},\mathrm{dec}\}.
\end{equation}
If the stopping criterion is not satisfied within $K_{\max}$ iterations, we set $K_q=K_{\max}$.
Thus, $K_{\mathrm{enc}}$ and $K_{\mathrm{dec}}$ are emergent computation depths determined by the corresponding implicit dynamics, SNR condition, input features, and convergence tolerance.

Due to parameter sharing across equilibrium iterations, the number of trainable parameters of the implicit encoder and decoder is independent of the effective depths $K_{\mathrm{enc}}$ and $K_{\mathrm{dec}}$. 
Thus, Implicit-JSCC has $O(1)$ parameter complexity with respect to the solver-determined effective depth, while its inference computation scales with the executed number of fixed-point iterations.
Algorithm~\ref{alg:eqsolve} summarizes the residual-based fixed-point solver used to obtain the equilibrium state and the corresponding effective depth.

\begin{algorithm}[t]
	\caption{$\mathrm{EqSolve}$: Residual-based Fixed-Point Solver}
	\label{alg:eqsolve}
	\begin{algorithmic}[1]
		
		\STATE \textbf{Input:} Implicit update $f$, conditioning feature $\mathbf{h}$, SNR embedding $\mathbf{e}_{\mathrm{snr}}$, stopping threshold $\tau$, maximum iteration number $K_{\max}$.
		\STATE \textbf{Output:} Final solver iterate $\mathbf{z}^{*}$ and effective depth $K$.
		
		\STATE Initialize $\mathbf{z}^{(0)}=\mathbf{0}$ and set $K=K_{\max}$.
		
		\FOR{$k=0,1,\ldots,K_{\max}-1$}
		\STATE Update $\mathbf{z}^{(k+1)}=f(\mathbf{z}^{(k)},\mathbf{h},\mathbf{e}_{\mathrm{snr}})$.
		\STATE Compute the mixed absolute--relative residual:
		\[
		r^{(k)}
		=
		\frac{
			\left\|
			\mathbf{z}^{(k+1)}
			-
			\mathbf{z}^{(k)}
			\right\|
		}{
			\max
			\left\{
			\left\|
			\mathbf{z}^{(k)}
			\right\|,
			1
			\right\}
		}.
		\]
		\IF{$r^{(k)}<\tau$}
		\STATE Set $K=k+1$ and stop iteration.
		\ENDIF
		\ENDFOR
		
		\STATE Set $\mathbf{z}^{*}$ as the final solver iterate.
		\STATE \textbf{return} $\mathbf{z}^{*},K$.
		
	\end{algorithmic}
\end{algorithm}

\subsection{Training Objective and Implicit Differentiation}

To improve the convergence stability of the implemented implicit blocks, we introduce a Jacobian regularization term during training.
For the encoder and decoder implicit updates, let
\begin{equation}
	\mathbf{J}_{q}
	=
	\left.
	\frac{\partial f_q}{\partial \mathbf{z}_{q}}
	\right|_{\mathbf{z}_{q}=\mathbf{z}_{q}^{*}},
	\quad
	q\in\{\mathrm{enc},\mathrm{dec}\},
\end{equation}
where $\mathbf{J}_{q}$ denotes the Jacobian of the implicit update with respect to the equilibrium state. 
Using random projection vectors $\mathbf{v}_{\mathrm{enc}}$ and $\mathbf{v}_{\mathrm{dec}}$ of compatible dimensions, the Jacobian regularization is written as

\begin{equation}
	\mathcal{L}_{\mathrm{Jac}}
	=
	\mathbb{E}_{\mathbf{v}_{\mathrm{enc}},\mathbf{v}_{\mathrm{dec}}}
	\left[
	\left\|
	\mathbf{J}_{\mathrm{enc}}\mathbf{v}_{\mathrm{enc}}
	\right\|_2^2
	+
	\left\|
	\mathbf{J}_{\mathrm{dec}}\mathbf{v}_{\mathrm{dec}}
	\right\|_2^2
	\right].
	\label{eq:jacobian_loss}
\end{equation}
The overall training objective is
\begin{equation}
	\mathcal{L}
	=
	\frac{1}{N}
	\left\|
	\hat{\mathbf{s}}
	-
	\mathbf{s}
	\right\|_2^2
	+
	\lambda_{\mathrm{Jac}}
	\mathcal{L}_{\mathrm{Jac}},
	\label{eq:training_objective}
\end{equation}
where $N$ denotes the total number of reconstructed elements in a mini-batch, and $\lambda_{\mathrm{Jac}}$ controls the strength of the stability regularization. 
The first term in \eqref{eq:training_objective} is the mean-squared error (MSE) reconstruction loss and serves as the main supervision signal for image reconstruction, while the Jacobian term suppresses unstable local amplification in the implicit dynamics and promotes fixed-point convergence.

During training, gradients through each equilibrium state are computed via implicit differentiation.
The backward linear solve is implemented through iterative vector-Jacobian products, avoiding storage of all forward fixed-point iterates.
For $q\in\{\mathrm{enc},\mathrm{dec}\}$, the adjoint variable $\mathbf{m}_q$ is obtained by solving
\begin{equation}
	\begin{aligned}
		\left(
		\mathbf{I}
		-
		\mathbf{J}_{q}^{\top}
		\right)
		\mathbf{m}_{q}
		&=
		\left(
		\frac{\partial \mathcal{L}}
		{\partial \mathbf{z}_{q}^{*}}
		\right)^{\top},
		\\
		\frac{\partial \mathcal{L}}
		{\partial \bm{\psi}_{q}}
		&=
		\mathbf{m}_{q}^{\top}
		\frac{
			\partial f_q(\mathbf{z}_{q}^{*},\mathbf{h}_{q},\mathbf{e}_{\mathrm{snr}};\bm{\psi}_{q})
		}{
			\partial \bm{\psi}_{q}
		}.
	\end{aligned}
	\label{eq:implicit_grad}
\end{equation}
This avoids explicitly backpropagating through all intermediate solver iterations and improves memory efficiency compared with unrolled deep networks.

\begin{algorithm}[t]
	\caption{Inference Procedure of Implicit-JSCC}
	\label{alg:implicit_jscc}
	\begin{algorithmic}[1]
		
		\STATE \textbf{Input:} Source image $\mathbf{s}$, channel SNR $\mathrm{SNR}_{\mathrm{dB}}$, maximum iteration number $K_{\max}$, thresholds $\tau_{\mathrm{enc}}$ and $\tau_{\mathrm{dec}}$.
		\STATE \textbf{Output:} Reconstructed image $\hat{\mathbf{s}}$, effective depths $K_{\mathrm{enc}}$ and $K_{\mathrm{dec}}$.
		
		\STATE Compute
		$\mathbf{e}_{\mathrm{snr}}
		=\mathrm{MLP}(
		\mathrm{SinEmbed}(\mathrm{SNR}_{\mathrm{dB}}))$.
		
		\vspace{0.2em}
		\STATE \textbf{Transmitter:}
		\STATE $\mathbf{h}_{\mathrm{enc}}=\mathrm{Stem}_{\mathrm{enc}}(\mathbf{s})$.
		\STATE $(\mathbf{z}_{\mathrm{enc}}^{*},K_{\mathrm{enc}})
		=
		\mathrm{EqSolve}
		(
		f_{\mathrm{enc}},
		\mathbf{h}_{\mathrm{enc}},
		\mathbf{e}_{\mathrm{snr}},
		\tau_{\mathrm{enc}},
		K_{\max}
		)$.
		\STATE $\mathbf{x}=\mathrm{Head}_{\mathrm{enc}}(\mathbf{z}_{\mathrm{enc}}^{*})$.
		
		\vspace{0.2em}
		\STATE \textbf{Wireless channel:}
		\STATE Convert $\mathbf{x}$ to $\mathbf{x}_{\mathrm c}$, transmit it according to~\eqref{eq:wireless_channel}, and obtain the real-valued received tensor $\mathbf{y}$.
		
		\vspace{0.2em}
		\STATE \textbf{Receiver:}
		\STATE $\mathbf{h}_{\mathrm{dec}}=\mathrm{InputProj}(\mathbf{y})$.
		\STATE $(\mathbf{z}_{\mathrm{dec}}^{*},K_{\mathrm{dec}})
		=
		\mathrm{EqSolve}
		(
		f_{\mathrm{dec}},
		\mathbf{h}_{\mathrm{dec}},
		\mathbf{e}_{\mathrm{snr}},
		\tau_{\mathrm{dec}},
		K_{\max}
		)$.
		\STATE $\hat{\mathbf{s}}=\mathrm{Head}_{\mathrm{dec}}(\mathbf{z}_{\mathrm{dec}}^{*})$.
		
		\STATE \textbf{return} $\hat{\mathbf{s}}$, $K_{\mathrm{enc}}$, $K_{\mathrm{dec}}$.
		
	\end{algorithmic}
\end{algorithm}

\section{Depth--SNR Analysis of Implicit Equilibrium Dynamics}
\label{sec:theory}

This section analyzes the central question of this paper:
\emph{given a channel SNR, what effective computation depth is required for semantic reconstruction?}
To answer this question, we develop an NNGP-inspired analytical framework for characterizing the depth--SNR behavior of the proposed Implicit-JSCC system.
It provides a tractable model for explaining how channel-induced perturbations affect equilibrium convergence.

\begin{remark}[\textbf{Receiver-side perturbation focus}]
	Since channel noise is injected between the encoder and decoder, it directly perturbs the decoder input representation.
	Therefore, the following kernel evolution analysis focuses on the decoder-side equilibrium dynamics, where the impact of channel perturbations on effective depth can be most directly characterized.
\end{remark}

\subsection{Theoretical Assumptions}

To obtain a tractable analytical characterization, we adopt the following assumptions. 

\begin{assumption}[\textbf{Infinite-width regime}]
The hidden feature dimension tends to infinity, such that the implicit transformation can be approximately characterized by an NNGP.
\end{assumption}

\begin{assumption}[\textbf{Locally linearized kernel dynamics}]
The representation features remain in a highly correlated regime, allowing the nonlinear kernel recursion to be locally approximated by linear contraction dynamics.
\end{assumption}

\begin{assumption}[\textbf{AWGN channel for theoretical analysis}]
For analytical tractability, the kernel evolution analysis focuses on the AWGN case.
\end{assumption}

\begin{assumption}[\textbf{Local effective kernel contraction}]
	The locally linearized kernel dynamics admit an effective kernel propagation factor $0<\beta<1$, which supports stable convergence of the linearized surrogate dynamics.
\end{assumption}

The factor $\beta$ is different from the feature-space Lipschitz constant $L_g$ in Proposition~1.
While $L_g$ characterizes the contraction property of the inner update $g_q(\cdot)$ in feature space, $\beta$ summarizes the effective propagation strength of kernel-space perturbations under the locally linearized surrogate dynamics.
In general, stronger feature-space contraction tends to induce faster covariance perturbation decay, suggesting that $\beta$ is related to the effective feature-space contraction strength.
Therefore, $\beta$ should be interpreted as a calibrated kernel-space propagation factor rather than a direct Lipschitz constant of the implemented neural block.

\subsection{Kernel Evolution of Implicit Equilibrium Iterations}

Consider the contraction-oriented implicit block introduced in Section~\ref{sec:block}:
\begin{equation}
	f(\mathbf{z})
	=
	(1-\alpha)\mathbf{z}
	+
	\alpha g(\mathbf{z}),
	\label{eq:damped_residual_kernel_form}
\end{equation}
where $\alpha\in(0,1)$ denotes the damping coefficient and $g(\cdot)$ is the inner nonlinear transformation.

Let
\begin{equation}
	\mathcal{K}^{(k)}(\mathbf{z},\mathbf{z}')
	=
	\mathbb{E}
	\left[
	z_i^{(k)}(\mathbf{z})
	z_i^{(k)}(\mathbf{z}')
	\right]
	\label{eq:kernel_definition}
\end{equation}
denote the NNGP covariance of representation features at equilibrium iteration $k$, which serves as a statistical surrogate for tracking representation-level perturbations.

Following the NNGP framework~\cite{lee2018deep}, the evolution of representation covariance can be described through an effective kernel-propagation perspective. For a standard infinitely wide nonlinear transformation, the corresponding kernel recursion takes the form
\begin{equation}
	\mathcal{K}^{(k+1)}
	=
	\sigma_b^2
	+
	\sigma_w^2
	F_{\phi}\left(\mathcal{K}^{(k)}\right),
	\label{eq:nngp_kernel_recursion}
\end{equation}
where $\sigma_w^2$ and $\sigma_b^2$ denote the weight and bias variances, respectively, and $F_{\phi}(\cdot)$ is determined by the activation function.

Let $\mathcal{K}^{\star}$ denote a local equilibrium kernel of
\eqref{eq:nngp_kernel_recursion}. Under the high-correlation assumption, the
nonlinear kernel propagation in \eqref{eq:nngp_kernel_recursion} is locally
approximated by a first-order expansion around $\mathcal{K}^{\star}$, which
motivates the following surrogate dynamics for the kernel deviation:
\begin{equation}
	\mathcal{K}^{(k+1)}
	-
	\mathcal{K}^{\star}
	=
	\mathcal{J}_{\star}
	\left[
	\mathcal{K}^{(k)}
	-
	\mathcal{K}^{\star}
	\right],
	\label{eq:linearized_kernel_recursion}
\end{equation}
where higher-order terms are neglected and $\mathcal{J}_{\star}$ denotes the local linearized kernel-propagation operator associated with \eqref{eq:nngp_kernel_recursion} at $\mathcal{K}^{\star}$.
For a tractable scalar characterization, the effective propagation factor $\beta$ in \textit{Assumption~4} is used to describe how much the kernel deviation from $\mathcal{K}^{\star}$ can remain
after one locally linearized iteration:
\begin{equation}
		\left\|
		\mathcal{J}_{\star}
		\left[
		\mathcal{K}^{(k)}
		-
		\mathcal{K}^{\star}
		\right]
		\right\|
		\leq
		\beta
		\left\|
		\mathcal{K}^{(k)}
		-
		\mathcal{K}^{\star}
		\right\|,
		\qquad
		0<\beta<1.
		\label{eq:effective_kernel_contraction}
\end{equation}
In a standard infinitely wide feed-forward network, this local kernel-propagation behavior is determined by the activation function and weight statistics.
For the implemented Implicit-JSCC system, however, the effective propagation behavior is also affected by the trained transceiver and solver configuration.
Therefore, we do not attempt to derive a closed-form expression for $\beta$ from individual architectural components.
Instead, $\beta$ is used as a calibrated effective kernel propagation factor of the trained system, which is estimated from validation-set effective-depth statistics.
A smaller $\beta$ indicates faster local convergence and stronger suppression of kernel-space perturbations.

\subsection{Channel Perturbation and Kernel Convergence}

Under the AWGN channel, the decoder input representation is modeled as
\begin{equation}
	\mathbf{Y}
	=
	\mathbf{Z}_{\rm clean}
	+
	\mathbf{N},
	\quad
	\mathbf{N}
	\sim
	\mathcal{N}(\mathbf{0},\sigma_n^2\mathbf{I}).
	\label{eq:decoder_awgn_perturbation}
\end{equation}
The injected channel noise perturbs the initial decoder-side kernel as
\begin{equation}
	\mathcal{K}_{\mathrm{noisy}}^{(0)}
	=
	\mathcal{K}_{\mathrm{clean}}^{(0)}
	+
	\sigma_n^2\mathbf{I}.
	\label{eq:noisy_kernel}
\end{equation}	
Define the initial kernel deviation from equilibrium as
\begin{equation}
	\Delta_0
	\triangleq
	\left\|
	\mathcal{K}_{\mathrm{noisy}}^{(0)}
	-
	\mathcal{K}^\star
	\right\|,
	\label{eq:initial_kernel_deviation_basic}
\end{equation}
where $\|\cdot\|$ denotes an operator norm in the kernel space.
The quantity $\Delta_0$ measures the initial representation perturbation induced by channel noise at the decoder input.

Within the locally linearized kernel model, the following theorem characterizes the decay of decoder-side channel perturbations.

\begin{theorem}[Kernel convergence under AWGN perturbation]
		\label{thm:kernel_convergence_awgn}
		Under the locally linearized surrogate dynamics in~\eqref{eq:linearized_kernel_recursion} and the effective contraction condition in~\eqref{eq:effective_kernel_contraction},
		for the decoder-side noisy initial kernel defined in~\eqref{eq:noisy_kernel} and the initial deviation $\Delta_0$ defined in~\eqref{eq:initial_kernel_deviation_basic}, the kernel deviation after $k$ equilibrium iterations satisfies
	\begin{equation}
		\left\|
		\mathcal{K}^{(k)}
		-
		\mathcal{K}^{\star}
		\right\|
		\leq
		\beta^k \Delta_0 .
		\label{eq:kernel_convergence_bound}
	\end{equation}
\end{theorem}

\textit{Proof.}
From the local surrogate dynamics in \eqref{eq:linearized_kernel_recursion}
and the effective contraction condition in
\eqref{eq:effective_kernel_contraction}, we have
\[
\left\|
\mathcal{K}^{(k+1)}
-
\mathcal{K}^{\star}
\right\|
=
\left\|
\mathcal{J}_{\star}
\!\left[
\mathcal{K}^{(k)}
-
\mathcal{K}^{\star}
\right]
\right\|
\leq
\beta
\left\|
\mathcal{K}^{(k)}
-
\mathcal{K}^{\star}
\right\|.
\]
Applying this inequality recursively from $0$ to $k-1$ gives
\[
\left\|
\mathcal{K}^{(k)}
-
\mathcal{K}^{\star}
\right\|
\leq
\beta
\left\|
\mathcal{K}^{(k-1)}
-
\mathcal{K}^{\star}
\right\|
\leq
\cdots
\leq
\beta^k
\left\|
\mathcal{K}^{(0)}
-
\mathcal{K}^{\star}
\right\|.
\]
For the noisy decoder initialization,
$\mathcal{K}^{(0)}=\mathcal{K}^{(0)}_{\mathrm{noisy}}$. Using the definition
of $\Delta_0$ in \eqref{eq:initial_kernel_deviation_basic}, we obtain
\[
\left\|
\mathcal{K}^{(k)}
-
\mathcal{K}^{\star}
\right\|
\leq
\beta^k
\left\|
\mathcal{K}^{(0)}_{\mathrm{noisy}}
-
\mathcal{K}^{\star}
\right\|
=
\beta^k\Delta_0.
\]
Since $0<\beta<1$, the decoder-side kernel perturbation decays
geometrically with the equilibrium iteration index $k$.
\hfill $\blacksquare$

This result provides an analytical explanation for why stronger channel perturbations require more receiver-side equilibrium iterations. Since a larger $\sigma_n^2$ generally induces a larger initial kernel deviation $\Delta_0$, lower-SNR channels require more iterations to reach the same perturbation tolerance.

\subsection{Depth--SNR Characterization}

\begin{corollary}[\textbf{Required effective depth}]
\label{cor:required_depth}
Let $\epsilon_{\mathrm{tol}}$ denote the target kernel convergence tolerance.
Under the conditions of Theorem~\ref{thm:kernel_convergence_awgn}, consider the iteration index $k$ required to satisfy
\begin{equation}
	\left\|\mathcal{K}^{(k)}-\mathcal{K}^\star\right\|
	\leq
	\epsilon_{\mathrm{tol}}.
	\label{eq:kernel_tolerance_requirement}
\end{equation}
Then, the corresponding sufficient continuous depth estimate is
\begin{equation}
	K_{\mathrm{dec}}^{\mathrm{pred}}
	=
	\frac{
		\log\left(\Delta_0/\epsilon_{\mathrm{tol}}\right)
	}{
		\log(1/\beta)
	}
	.
	\label{eq:required_depth_estimate}
\end{equation}
\end{corollary}

\textit{Proof.}
From Theorem~\ref{thm:kernel_convergence_awgn}, a sufficient condition for \eqref{eq:kernel_tolerance_requirement} is
\begin{equation}
	\beta^{k}\Delta_0
	\leq
	\epsilon_{\mathrm{tol}} .
	\label{eq:depth_condition_start}
\end{equation}
Assuming $\Delta_0>\epsilon_{\mathrm{tol}}$, dividing both sides by $\Delta_0$ gives
\begin{equation}
	\beta^{k}
	\leq
	\frac{\epsilon_{\mathrm{tol}}}{\Delta_0}.
\end{equation}
Taking logarithms on both sides yields
\begin{equation}
	k\log\beta
	\leq
	\log\left(
	\frac{\epsilon_{\mathrm{tol}}}{\Delta_0}
	\right).
\end{equation}
Since $0<\beta<1$, we have $\log\beta<0$.
Dividing by $\log\beta$ therefore reverses the inequality:
\begin{equation}
	k
	\geq
	\frac{
		\log\left(\epsilon_{\mathrm{tol}}/\Delta_0\right)
	}{
		\log\beta
	}
	=
	\frac{
		\log\left(\Delta_0/\epsilon_{\mathrm{tol}}\right)
	}{
		\log(1/\beta)
	}
	\equiv
	K_{\mathrm{dec}}^{\mathrm{pred}}.
\end{equation}
This continuous iteration threshold gives \eqref{eq:required_depth_estimate}.
\hfill $\blacksquare$

Since the signal power is normalized, the AWGN noise variance satisfies
$\sigma_n^2=10^{-\mathrm{SNR}_{\mathrm{dB}}/10}$. 
Motivated by the additive perturbation form in~\eqref{eq:noisy_kernel}, we adopt the following first-order parameterization of the initial decoder-side kernel deviation:
\begin{equation}
	\Delta_0(\mathrm{SNR})
	=
	c_0+c_1 10^{-\mathrm{SNR}_{\mathrm{dB}}/10},
	\label{eq:delta0_snr_parameterization}
\end{equation}
where $c_0$ captures the residual high-SNR deviation caused by compression or modeling error, and $c_1$ controls the contribution of channel noise.
Substituting \eqref{eq:delta0_snr_parameterization} into \eqref{eq:required_depth_estimate} yields the continuous effective-depth estimate:
\begin{equation}
		K_{\mathrm{dec}}^{\mathrm{pred}}(\mathrm{SNR})
		=
		\frac{
			\log\left[
			\left(
		c_0+c_1 10^{-\mathrm{SNR}_{\mathrm{dB}}/10}
		\right)
		/\epsilon_{\mathrm{tol}}
		\right]
		}{
			\log(1/\beta)
		},
	\label{eq:depth_snr_relation}
\end{equation}
which defines a continuous
analytical model for the theory-guided effective depth. 
This relation indicates that lower SNR enlarges the initial perturbation and
therefore increases the required effective depth, with a logarithmic dependence
on the perturbation magnitude.

\begin{remark}[\textbf{Nontrivial perturbation regime}]
The depth estimate in \eqref{eq:required_depth_estimate} is mainly used in the
nontrivial regime where the initial kernel deviation is larger than the target
tolerance, i.e., $\Delta_0>\epsilon_{\mathrm{tol}}$. If
$\Delta_0\leq\epsilon_{\mathrm{tol}}$, the kernel deviation already satisfies
the convergence requirement at initialization in the idealized kernel model, and
no additional equilibrium refinement is required. In practical solvers, however,
a minimum number of iterations may still be performed due to implementation
constraints and numerical stabilization.
\end{remark}

\begin{remark}[\textbf{Kernel depth versus solver residual depth}]
	The predicted depth $K_{\mathrm{dec}}^{\mathrm{pred}}$ derived above is defined in terms of kernel-space perturbation convergence and gives a continuous prediction of the mean iterations required for the kernel deviation $\|\mathcal{K}^{(k)}-\mathcal{K}^\star\|$ to reach a prescribed tolerance $\epsilon_{\mathrm{tol}}$.
	In contrast, the decoder effective depth observed during inference, denoted by $K_{\mathrm{dec}}$, is determined by the feature-space residual stopping criterion, namely the relative difference between two consecutive solver iterates.
	
	Therefore, $K_{\mathrm{dec}}^{\mathrm{pred}}$ and $K_{\mathrm{dec}}$ are not strictly identical quantities.
	The former is a theory-guided depth under the locally linearized kernel dynamics, whereas the latter is an implementation-level solver iteration count measured from finite-width feature trajectories. 
	In this paper, $K_{\mathrm{dec}}^{\mathrm{pred}}$ is interpreted as the predicted mean operating depth for adaptive equilibrium refinement, rather than an exact per-sample solver iteration count or a reconstruction-quality-optimal depth.
\end{remark}

\begin{algorithm}[!t]
	\caption{Offline Calibration of the Theory-Guided Depth--SNR Model}
	\label{alg:depth_snr_fitting}
	\begin{algorithmic}[1]
		\STATE \textbf{Input:} Trained Implicit-JSCC, calibration set $\mathcal{V}$, SNR grid $\mathcal{S}$,
		high-SNR subset $\mathcal{S}_{\mathrm{H}}=\{s\in\mathcal{S}:s\geq5~\mathrm{dB}\}$, and solver thresholds.
		\STATE Set $\epsilon_{\mathrm{tol}}=\tau_{\mathrm{dec}}$.
		\FOR{each $s\in\mathcal{S}$}
		\STATE Run inference on $\mathcal{V}$ and measure
		$\overline{K}_{\mathrm{dec}}(s)$.
		\ENDFOR
		\STATE Compute $\overline{K}_{\mathrm{H}}$ from
		$\{\overline{K}_{\mathrm{dec}}(s):
		s\in\mathcal{S}_{\mathrm{H}}\}$.
		\STATE Solve \eqref{eq:depth_snr_fitting} by constrained nonlinear
		least squares.
		\STATE Generate the continuous depth--SNR map using
		\eqref{eq:depth_snr_relation}.
		\STATE \textbf{return}
		$\hat{\beta}$, $\hat{\eta}$, $\hat{c}_0$, $\hat{c}_1$ and the calibrated theory-guided depth--SNR map.
	\end{algorithmic}
\end{algorithm}

It is worth emphasizing that the proposed depth--SNR relation is a theory-guided parametric model rather than a purely empirical regression curve. The kernel perturbation analysis determines its functional form, including the logarithmic dependence on the initial channel-induced perturbation and the effective contraction factor. The measured effective-depth statistics are used only to calibrate the system-specific parameters of the trained finite-width Implicit-JSCC model, including the effective propagation factor and relative perturbation scale.

Such calibration is necessary because the effective equilibrium dynamics depend not only on the wireless channel condition but also on the source distribution and the neural transceiver architecture. The combined effects of finite-width implementation, SNR-conditioned neural processing, and approximate fixed-point solving cannot be determined from the idealized model alone. The resulting depth--SNR map should therefore be interpreted as a calibrated theory-guided characterization of the mean operating depth for the trained system, rather than as an unconstrained empirical fit or a data-independent prediction for arbitrary transceivers.

Algorithm~\ref{alg:depth_snr_fitting} summarizes the offline calibration of the theory-guided depth--SNR model in~\eqref{eq:depth_snr_relation}.
The calibration is performed once after training on a held-out validation set and is not involved in online inference.
During actual transmission, the decoder still determines its effective depth through the residual-based stopping rule.
First, the trained Implicit-JSCC is evaluated on the calibration set to measure the mean decoder effective depth at each SNR.
Second, the mean high-SNR depth $\overline{K}_{\mathrm{H}}$ is used to anchor the noise-independent term $c_0$, because the noise-dependent term vanishes as the SNR increases. Finally, $\beta$ and the positive ratio $\eta=c_1/c_0$ are estimated by fitting the analytical prediction $K_{\mathrm{dec}}^{\mathrm{pred}}(s)$ to the measured depths. We set
$c_0=\epsilon_{\mathrm{tol}} (1/\beta)^{\overline{K}_{\mathrm{H}}}$ and $c_1=\eta c_0$, and fit the two remaining variables using
\begin{equation}
	(\hat{\beta},\hat{\eta})
	=
	\underset{\substack{0<\beta<1\\ \eta>0}}{\arg\min}
	\sum_{s\in\mathcal{S}}
	\left[
	K_{\mathrm{dec}}^{\mathrm{pred}}(s;\beta,\eta)
	-\overline{K}_{\mathrm{dec}}(s)
	\right]^2 .
	\label{eq:depth_snr_fitting}
\end{equation}
Here, $\mathcal{S}$ denotes the set of SNR values used for calibration, and
$\overline{K}_{\mathrm{dec}}(s)$ is the mean measured decoder effective
depth on the calibration set at SNR $s$.
$K_{\mathrm{dec}}^{\mathrm{pred}}(s;\beta,\eta)$ denotes the continuous prediction
in \eqref{eq:depth_snr_relation} after substituting the above expressions
for $c_0$ and $c_1$. The calibrated parameters are recovered as
$\hat{c}_0=\epsilon_{\mathrm{tol}}
(1/\hat{\beta})^{\overline{K}_{\mathrm{H}}}$ and
$\hat{c}_1=\hat{\eta}\hat{c}_0$.

We set $\epsilon_{\mathrm{tol}}=\tau_{\mathrm{dec}}$ as a normalized
reference. This is a normalization convention for fitting and
does not assert that the kernel-space and feature-space stopping criteria are
physically identical. Since $c_0$ and $c_1$ are identifiable only relative
to this tolerance, we report $c_0/\epsilon_{\mathrm{tol}}$ and $c_1/c_0$.
For the AWGN setting, the calibration is performed on the independent DIV2K validation set and gives $\overline{K}_{\mathrm{H}}=8$, $\hat{\beta}=0.2138$, and $\hat{\eta}=0.2832$. This corresponds to $c_0/\epsilon_{\mathrm{tol}}=2.29\times10^5$ and $c_1/c_0=0.2832$.
These values are fixed when generating the continuous
depth--SNR curve.

\begin{table}[t!]
	\centering
	\caption{Implementation parameters}
	\label{tab:implementation_settings}
	\renewcommand{\arraystretch}{1.1}
	\begin{tabular}{@{}ll@{}}
		\toprule
		Item & Setting \\
		\midrule
		Hidden dimension & 192 \\
		SNR embedding dimension & 32 \\
		Stem channels & 64 \\
		Rician factor & $\kappa=3$ \\
		
		Maximum solver depth & $K_{\max}=30$ \\
		Forward solver & Plain fixed-point iteration \\
		Stopping thresholds & $\tau_{\mathrm{enc}}=0.08$, $\tau_{\mathrm{dec}}=0.05$ \\
		
		Optimizer / loss & AdamW / MSE + Jacobian regularization \\
		Jacobian regularization weight & $\lambda_{\mathrm{Jac}}=5\times10^{-4}$\\
		Learning rate & Initial/final $3\times10^{-5}$ to $5\times10^{-7}$ \\
		Weight decay / grad. clip & $10^{-4}$ / 0.5 \\
		Batch size / epochs & 8 / 180 \\
		\bottomrule
	\end{tabular}
\end{table}

\section{Experimental Results}\label{sec:exp}

The proposed Implicit-JSCC model is trained on DIV2K~\cite{agustsson2017ntire} using randomly cropped $256\times256$ image patches and is evaluated on the Kodak24 dataset~\cite{kodak1993}.
The model is trained from scratch under each matched channel setting. We use a 25-epoch SNR curriculum that gradually expands the lower SNR bound from 0 dB to $-10$ dB, followed by full-range SNR sampling with high-SNR oversampling over $[15,30]$~dB to improve high-SNR reconstruction. The remaining implementation details are summarized in Table~\ref{tab:implementation_settings}.

We compare the proposed Implicit-JSCC with three learned JSCC baselines and two separated digital baselines. The separated baselines combine Better Portable Graphics (BPG) or Joint Photographic Experts Group (JPEG) source coding with low-density parity-check (LDPC) channel coding and adaptive modulation and coding (AMC). For each SNR, all candidate modulation and coding schemes (MCSs) are evaluated through bit-level channel simulation and decodability testing under the same complex-symbol budget corresponding to $R=1/12$. 
The highest-quality decodable operating point is then used to determine the source-bit budget for evaluating all Kodak images. If no candidate MCS can successfully recover the complete source bitstream, the corresponding PSNR and SSIM values are set to zero. Therefore, the separated-baseline curves should be interpreted as best-decodable performance envelopes under the same channel-use constraint, rather than as fixed online AMC policies.

\begin{itemize}
	\item \textbf{DeepJSCC~\cite{bourtsoulatze2019deep}:} a representative learned joint source-channel coding baseline for wireless image transmission. For a fair comparison, we train one DeepJSCC model for each channel type using the same image size, bandwidth ratio, and training SNR set as Implicit-JSCC.
	
	\item \textbf{SwinJSCC~\cite{yang2025swinjscc}:} a Swin Transformer-based JSCC baseline. We retrain SwinJSCC under the same image size and bandwidth ratio. For each channel, we follow the two-stage training protocol: \texttt{SwinJSCC\_w/o\_SAandRA} is first trained at 10 dB, and the resulting checkpoint is then used to initialize the SNR-adaptive \texttt{SwinJSCC\_w/\_SA} model trained over the same multi-SNR set.
	
	\item \textbf{ADJSCC~\cite{xu2022adjscc}:} an attention-based adaptive JSCC baseline. We use the official attention encoder and decoder and train channel-matched models for AWGN, Rayleigh, and Rician channels. The SNR is provided to the attention modules and the channel layer, following the adaptive design of ADJSCC.
	
	\item \textbf{BPG+LDPC+AMC:} a separated source-channel coding baseline using BPG as the source codec. The complex-symbol budget is matched to the same bandwidth ratio $R=1/12$. We use an LDPC block length of 768, candidate code rates $1/2$, $2/3$, and $3/4$, and unit-power binary phase-shift keying (BPSK), quadrature phase-shift keying (QPSK), 16-ary quadrature amplitude modulation (16-QAM), and 64-ary quadrature amplitude modulation (64-QAM).
	
	\item \textbf{JPEG+LDPC+AMC:} a separated baseline evaluated using the same LDPC simulation, modulation--coding search, channel-use constraint, and decodability criterion as BPG+LDPC+AMC, with JPEG used as the source codec.
	
\end{itemize}

\subsection{Overall Reconstruction Performance}

\begin{figure*}[!t]
	\centering
	\includegraphics[width=0.92\textwidth]{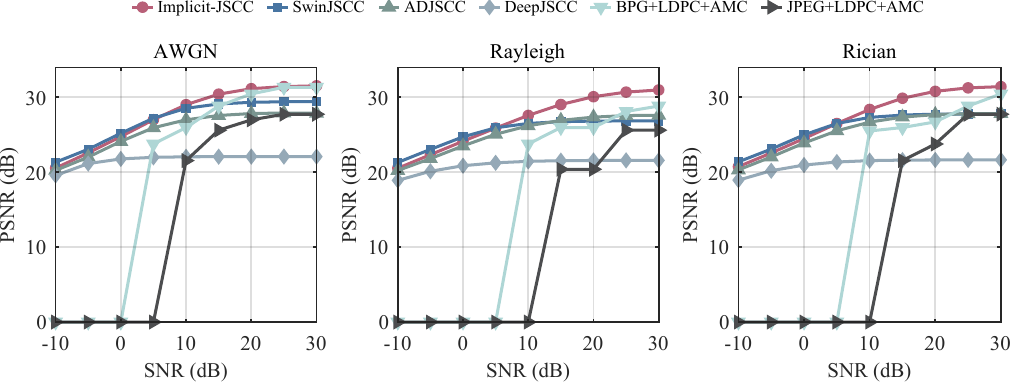}
	\caption{PSNR performance comparison under AWGN, Rayleigh, and Rician channels with the same bandwidth ratio $R=1/12$. The BPG+LDPC+AMC and JPEG+LDPC+AMC curves denote best-decodable MCS envelopes obtained by exhaustive bit-level evaluation of BPSK, QPSK, 16-QAM, and 64-QAM with three LDPC rates.}
	\label{fig:algorithm_comparison_psnr}
\end{figure*}

\begin{figure*}[!t]
	\centering
	\includegraphics[width=0.92\textwidth]{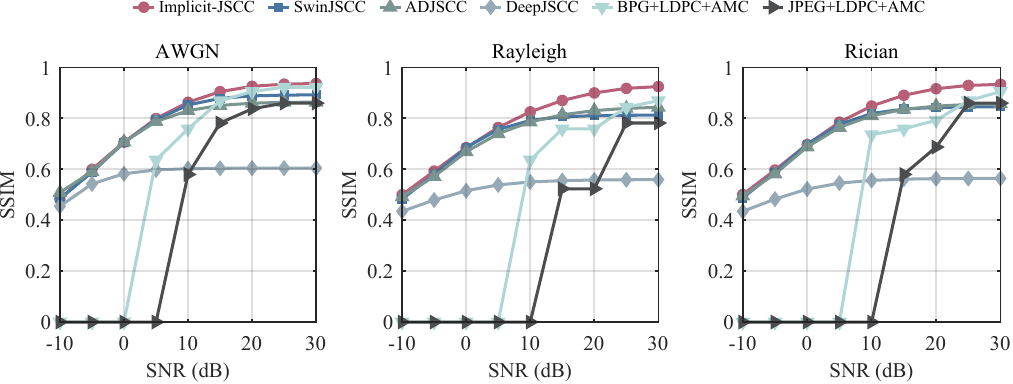}
	\caption{SSIM performance comparison under AWGN, Rayleigh, and Rician channels with the same bandwidth ratio $R=1/12$. The separated baselines use the same best-decodable MCS envelope and outage definition as in Fig.~\ref{fig:algorithm_comparison_psnr}.}
	\label{fig:algorithm_comparison_ssim}
\end{figure*}

Fig.~\ref{fig:algorithm_comparison_psnr} compares the PSNR performance under AWGN, Rayleigh, and Rician channels. 
Several observations can be made. 
First, all learned JSCC methods exhibit graceful degradation as the SNR decreases, whereas the best-decodable BPG+LDPC+AMC and JPEG+LDPC+AMC envelopes show threshold-like behavior and outage regions. 
This is because the separated schemes rely on successful channel decoding before image reconstruction, while neural JSCC maps channel perturbations directly into reconstruction distortion. 
Second, Implicit-JSCC maintains stable reconstruction quality across all three channel models, indicating that the implicit equilibrium decoder can refine noisy channel observations not only in AWGN but also under fading-induced perturbations. 
Third, compared with explicit neural baselines, Implicit-JSCC achieves competitive PSNR while using a compact parameter budget and supporting solver-determined adaptive computation. 
This suggests that the proposed method does not simply improve performance by increasing backbone capacity; instead, it benefits from equilibrium refinement and parameter sharing across iterations.

Fig.~\ref{fig:algorithm_comparison_ssim} reports the corresponding SSIM results. 
The SSIM curves show a trend consistent with the PSNR comparison: Implicit-JSCC preserves structural fidelity over a wide SNR range and remains robust under fading channels. 
The competitive SSIM performance at medium and high SNRs indicates that the implicit decoder does not merely reduce pixel-wise distortion, but also helps recover coherent image structures after channel corruption. 
The separated-baseline envelopes again exhibit abrupt transitions because a small change in SNR can determine whether a higher-throughput modulation--coding combination becomes exactly decodable. 
In contrast, Implicit-JSCC produces continuous quality adaptation, which is desirable for wireless semantic transmission under gradually varying channel conditions.

\subsection{Adaptive Depth Control}

The implicit receiver enables inference-time computation--quality control through the stopping threshold. To evaluate this property, we vary the decoder stopping threshold $\tau_{\mathrm{dec}}$ while keeping the trained network unchanged. A smaller $\tau_{\mathrm{dec}}$ imposes a stricter terminal-residual criterion and generally increases the number of decoder iterations, whereas a larger $\tau_{\mathrm{dec}}$ permits earlier termination. Because the solver residual and reconstruction distortion are distinct quantities, reconstruction quality need not vary monotonically with the stopping threshold.
\begin{figure}[!t]
	\centering
	\includegraphics[width=\columnwidth]{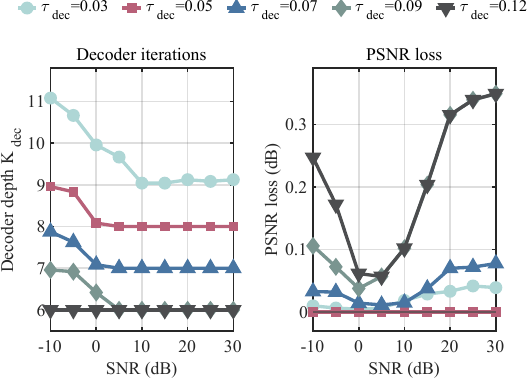}
		\caption{Receiver stopping-threshold tradeoff under the AWGN channel. Left: average decoder effective depth. Right: PSNR loss relative to $\tau_{\mathrm{dec}}=0.05$.}
	\label{fig:awgn_decoder_tau_tradeoff}
\end{figure}

Fig.~\ref{fig:awgn_decoder_tau_tradeoff} presents the decoder effective depth and the corresponding PSNR loss under different stopping thresholds. The strict threshold $\tau_{\mathrm{dec}}=0.03$ requires the largest number of decoder iterations but produces a slight PSNR loss of less than $0.042$ dB relative to $\tau_{\mathrm{dec}}=0.05$, showing that additional equilibrium iterations do not necessarily improve reconstruction quality. Relaxing the threshold to $0.07$ reduces the effective depth with less than $0.078$ dB loss, whereas $0.09$ and $0.12$ can incur losses of approximately $0.349$ dB. The default $\tau_{\mathrm{dec}}=0.05$ therefore provides the best observed fidelity--complexity balance among the evaluated thresholds.

These results demonstrate that Implicit-JSCC supports receiver-side computation--quality control without retraining. 
All stopping-threshold results are obtained using the same trained checkpoint, and only the residual tolerance is changed during inference. 
This differs from conventional fixed-depth JSCC architectures, whose computational cost is largely determined once the network depth is selected. 
Thus, the stopping threshold controls the terminal residual and computation: stricter stopping increases computation without guaranteeing a PSNR gain, whereas moderate relaxation reduces computation with a limited reconstruction penalty.

\subsection{Theory-Guided Depth Characterization}

\begin{figure*}[!t]
	\centering
	\includegraphics[width=0.82\textwidth]{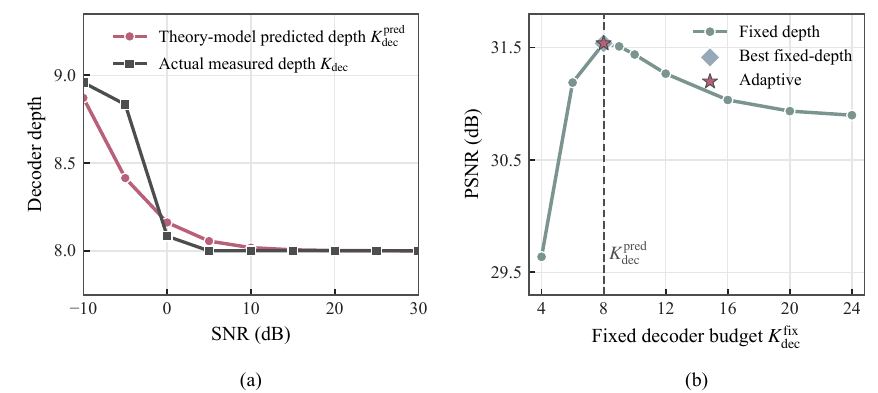}
	\caption{Calibrated theory-guided decoder effective-depth characterization and decoder-only fixed-depth validation under the AWGN channel. Left: continuous analytical prediction calibrated on the independent DIV2K validation set and measured decoder depths on Kodak24. Right: paired fixed-decoder-depth inference sweep at 30 dB, with the encoder retaining adaptive stopping.}
	\label{fig:theory_required_depth}
\end{figure*}

Fig.~\ref{fig:theory_required_depth} further examines whether the proposed depth--SNR analysis provides useful guidance for solver configuration.
The theory-guided depth--SNR characterization is obtained through offline calibration of the system-specific parameters in \eqref{eq:depth_snr_relation} using the effective-depth statistics measured on the independent DIV2K validation set, following Algorithm~\ref{alg:depth_snr_fitting}.
These calibrated parameters are used only for depth--SNR characterization, while the decoder still follows the residual stopping rule during actual inference.

The left panel compares this validation-calibrated continuous prediction with the adaptive decoder depths measured on Kodak24.
The theory-guided estimate decreases as SNR increases, which matches the measured trend of adaptive decoder iterations. 
This agreement shows that the proposed parametric form captures the observed depth--SNR trend and is consistent with the theoretical expectation that lower SNR produces a larger initial perturbation and requires more equilibrium refinement.
The right panel further illustrates the practical relevance of the predicted operating region through decoder-only fixed-depth inference at 30 dB.
When the fixed budget is too small, the decoder terminates before sufficient representation refinement is achieved, leading to a clear PSNR loss. 
As the budget increases toward the predicted depth region, the reconstruction quality improves rapidly. 
At 30 dB, the continuous theory-model prediction obtained from the validation-calibrated parameters is $K_{\mathrm{dec}}^{\mathrm{pred}}=8.0002$, which is practically identical to the measured adaptive depth.
Moreover, the predicted depth lies within the near-saturation region observed in the fixed-depth sweep.
The corresponding PSNR value at fixed depth 8 is 31.54 dB. Thus, the calibrated continuous prediction is consistent with the observed operating depth. 
Further increasing the number of iterations does not guarantee better reconstruction quality.

\begin{figure}[!t]
	\centering
	\includegraphics[width=0.94\columnwidth]{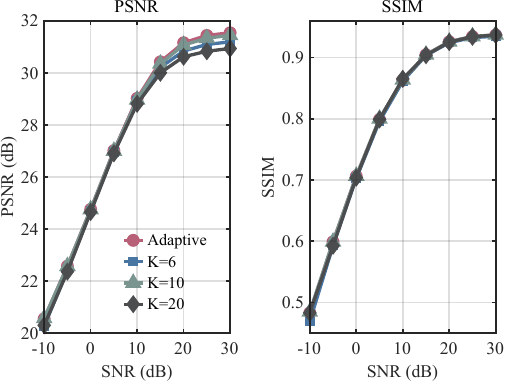}
	\caption{Decoder-only fixed-depth inference ablation under the AWGN channel.}
	\label{fig:fixed_k_ablation}
\end{figure}

Fig.~\ref{fig:fixed_k_ablation} highlights the limitation of using a manually fixed decoder depth. 
A small fixed budget, such as $K_{\mathrm{dec}}^{\mathrm{fix}}=6$, can be insufficient because the decoder is forced to stop before the residual dynamics reach the desired tolerance. 
This effect becomes especially visible at medium and high SNRs, where the channel noise is no longer the dominant bottleneck and reconstruction quality becomes more sensitive to the completeness of semantic refinement. 
A moderate budget such as $K_{\mathrm{dec}}^{\mathrm{fix}}=10$ approaches the adaptive result, but it still wastes computation on samples or SNRs that would have converged earlier. 
At 30 dB, $K_{\mathrm{dec}}^{\mathrm{fix}}=20$ is 0.60 dB below adaptive stopping. This decoder-only control confirms the trend observed in the dense sweep while excluding simultaneous changes in encoder depth.
These results indicate the advantage of adaptive stopping, which assigns computation according to each received sample's residual trajectory instead of using a universal fixed decoder budget.

\begin{figure}[!t]
	\centering
	\includegraphics[width=0.94\columnwidth]{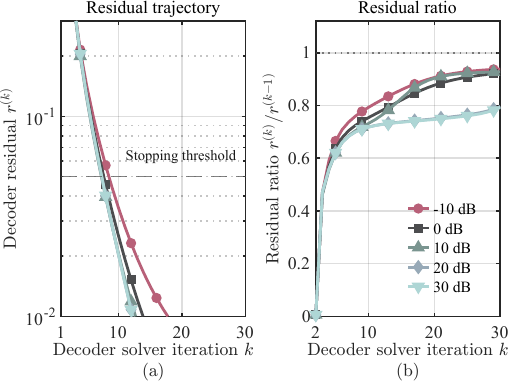}
	\caption{Residual convergence behavior of the implicit decoder under AWGN channel. Left: decoder residual trajectories with the decoder stopping threshold $\tau_{\mathrm{dec}}=0.05$. Right: empirical residual ratios.}
	\label{fig:contraction_residual_awgn}
\end{figure}
To verify the convergence behavior of the proposed implicit receiver, we record the residual trajectory of the decoder fixed-point solver during inference. 
We also report the empirical residual ratio $r^{(k)}/r^{(k-1)}$, which provides an implementation-level diagnostic of whether the fixed-point residual decreases across iterations.
As shown in Fig.~\ref{fig:contraction_residual_awgn}, the decoder residual decreases steadily with the solver iteration index $k$ and eventually falls below the stopping threshold. 
The empirical residual ratio is mostly below one after a short transient stage, indicating that each solver step reduces the discrepancy between consecutive fixed-point iterates. 
This confirms that the learned implicit decoder behaves as a stable refinement process rather than producing divergent or oscillatory updates. 
Lower-SNR inputs require slightly more iterations to reach the same tolerance because stronger channel perturbations lead to larger initial deviations, whereas high-SNR inputs converge more rapidly.
This behavior is consistent with the contractive-dynamics assumption used in the depth--SNR analysis. 
Although the theory characterizes perturbation decay in kernel space, the measured feature-space residuals exhibit a similar approximately geometric decay.
Therefore, the residual trajectories provide an implementation-level diagnostic of equilibrium convergence and are consistent with the contractive refinement behavior assumed in the analytical model. They are analyzed separately and are not used to estimate the kernel propagation factor $\beta$, which is calibrated from the measured effective-depth--SNR relationship.

\subsection{Complexity and Resource Analysis}

\begin{table*}[!t]
	\centering
	\caption{Model complexity comparison of Implicit-JSCC variants and baseline schemes. Peak memory is reported in mebibytes (MiB)}
	\label{tab:complexity_comparison}
	\renewcommand{\arraystretch}{1.1}
	\begin{tabular}{@{}>{\raggedright\arraybackslash}p{0.18\textwidth}cc>{\raggedright\arraybackslash}p{0.34\textwidth}>{\raggedright\arraybackslash}p{0.16\textwidth}@{}}
		\toprule
		Method & Trainable parameters & Peak memory & Encoder/decoder blocks & Depth policy \\
		\midrule
		Implicit-JSCC & 2.42M & 70.05 MiB & shared DEQ (encoder/decoder), $K_{\mathrm{enc}}/K_{\mathrm{dec}}=9.0/8.0$ & Adaptive fixed-point \\
		MLP-Implicit-JSCC variant & 1.06M & 67.17 MiB & shared MLP (encoder/decoder), $K_{\mathrm{enc}}/K_{\mathrm{dec}}=14.0/16.3$ & Adaptive fixed-point \\
		SwinJSCC & 21.89M & 226.73 MiB & 8 Swin / 8 Swin & Explicit feed-forward \\
		ADJSCC & 10.66M & -- & 5 convolution-attention / 5 convolution-attention & Explicit feed-forward \\
		DeepJSCC & 0.14M & 13.96 MiB & 5 convolution / 5 deconvolution & Explicit feed-forward \\
		BPG+LDPC+AMC & -- & -- & codec / LDPC / AMC & Best-decodable MCS \\
		JPEG+LDPC+AMC & -- & -- & codec / LDPC / AMC & Best-decodable MCS \\
		\bottomrule
	\end{tabular}
\end{table*}

Table~\ref{tab:complexity_comparison} summarizes the structural complexity of the compared schemes. The complexity reported for Implicit-JSCC corresponds to the Residual Conv-DEQ instantiation used in the main experiments. It uses one shared fixed-point block in the encoder and one in the decoder, while the actual computation is determined by the residual stopping rule. 
Since the shared fixed-point block is reused across iterations, the dominant inference cost scales approximately linearly with the executed solver depth.
Therefore, the reported $K_{\mathrm{enc}}$ and $K_{\mathrm{dec}}$ should be interpreted as
executed solver depths rather than independently parameterized layers. Compared with SwinJSCC, Implicit-JSCC uses substantially fewer trainable parameters and a lower peak memory footprint. The memory is measured on an NVIDIA GeForce RTX 5080 with PyTorch CUDA inference using a $256\times256$ RGB input. The SwinJSCC row uses the matched small model with depths $[2,2,2,2]$ for both encoder and decoder. The MLP-based implicit variant is included to illustrate that the proposed framework can use a lighter fixed-point map. It reduces the parameter count to 1.06M and the peak memory to 67.17 MiB, and its best AWGN result reaches 31.50 dB PSNR and 0.94 SSIM at 30 dB. However, this lightweight variant requires more fixed-point iterations and provides a weaker spatial image prior. Hence, the Residual Conv-DEQ variant is adopted as the main performance-oriented instantiation of the proposed framework. DeepJSCC has the smallest memory footprint because it adopts a lightweight convolutional neural network (CNN) backbone, but its reconstruction performance is much lower under the same bandwidth ratio. The
peak-memory entries are omitted for ADJSCC and the traditional BPG/JPEG+LDPC+AMC baselines because their implementations use different software stacks and codec libraries. For the same reason, raw runtime is not used as a primary metric. Table~\ref{tab:complexity_comparison} instead reports trainable parameters, peak GPU memory, and executed solver depths as implementation-level structural indicators.

\begin{figure*}[!t]
	\centering
	\includegraphics[width=0.97\textwidth]{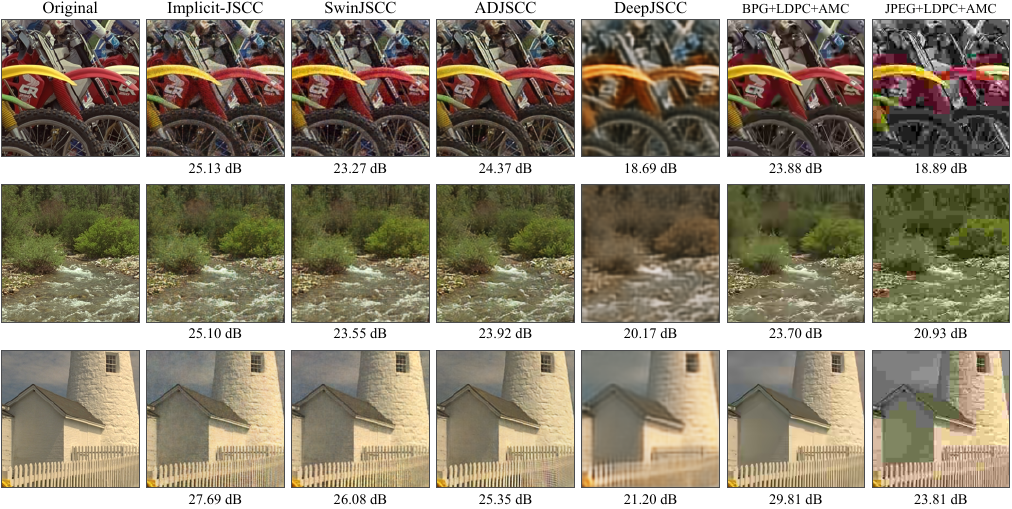}
	\caption{Visual reconstruction comparison under the AWGN channel at $\mathrm{SNR}=10$ dB. For each image, the BPG+LDPC+AMC and JPEG+LDPC+AMC columns report the highest-PSNR operating point among all exactly decoded combinations of four modulation orders and three LDPC rates.}
	\label{fig:resconvdeq_reconstruction_grid}
\end{figure*}

\subsection{Visual Reconstruction Results}
Fig.~\ref{fig:resconvdeq_reconstruction_grid} compares the visual reconstruction quality of different transmission schemes under the AWGN channel at $\mathrm{SNR}=10$ dB, representing a moderate-SNR operating condition. Three Kodak images are selected to cover different visual contents, including structured objects, natural textures, and smooth architectural regions. All methods are evaluated under the same bandwidth ratio and the same channel condition. The PSNR value shown below each reconstructed image is computed with respect to the corresponding original image. The BPG+LDPC+AMC and JPEG+LDPC+AMC results are obtained by exhaustive bit-level evaluation of all modulation--coding combinations for each image, retaining only operating points that recover the complete source bitstream. Across the three examples, Implicit-JSCC achieves an average PSNR/SSIM of 25.97~dB/0.8115, compared with 25.80~dB/0.7178 for BPG+LDPC+AMC and 21.21~dB/0.5491 for JPEG+LDPC+AMC. The proposed Implicit-JSCC preserves the main object contours, global color tone, and local structural details across the selected examples. Compared with DeepJSCC, it produces clearer boundaries and alleviates over-smoothed texture patterns. Compared with SwinJSCC and ADJSCC, Implicit-JSCC achieves competitive visual fidelity while additionally supporting residual-based adaptive fixed-point inference. The separated BPG+LDPC+AMC and JPEG+LDPC+AMC baselines can provide reasonable reconstructions when channel decoding succeeds, but they are more sensitive to the source-coding and channel-decoding operating point, leading to blocking artifacts, texture loss, or local structural distortions under the same channel-use budget. These qualitative observations are consistent with the PSNR and SSIM comparisons in Fig.~\ref{fig:algorithm_comparison_psnr} and Fig.~\ref{fig:algorithm_comparison_ssim}. They further show that Implicit-JSCC maintains stable perceptual reconstruction quality while retaining the effective-depth inference capability provided by the implicit equilibrium decoder.

\section{Conclusion}\label{s6}

This paper proposed Implicit-JSCC, an implicit equilibrium framework for deep joint source-channel coding that determines the effective computation depth through fixed-point solver convergence.
By reformulating semantic encoding and decoding as fixed-point equilibrium processes, the proposed method avoids manually predefined encoder and decoder depths while sharing parameters across equilibrium iterations.
An NNGP-inspired kernel evolution analysis was developed to relate channel-induced representation perturbations and the continuous receiver-side refinement depth.
Experiments on DIV2K and Kodak24 demonstrated the effectiveness of the proposed framework.
The results show that Implicit-JSCC achieves stable reconstruction performance and supports receiver-side computation--quality control through residual-based adaptive stopping.
The theoretical analysis explains why larger initial perturbations require more iterations to reach a prescribed convergence tolerance, while the fixed-depth simulations show that insufficient iterations can lead to under-refinement and that excessive iterations do not necessarily improve reconstruction quality.
These results suggest that implicit equilibrium modeling provides a promising approach for studying and controlling effective depth in semantic communication systems.

\bibliographystyle{IEEEtran}
\bibliography{Refer_ImplicitJSCC}

\end{document}